\documentclass[twoside,11pt]{article} 
\usepackage[a4paper,margin=0.82in]{geometry}

\usepackage{times}
\usepackage{amsthm}
\usepackage{thmtools}
\usepackage{thm-restate}
\usepackage{amsmath}
\usepackage[utf8]{inputenc} 
\usepackage[T1]{fontenc}    
\usepackage{tikz}
\usepackage{url}
\usepackage{float} 
\usepackage{booktabs}       
\usepackage{amsfonts}       
\usepackage{nicefrac}       
\usepackage{microtype}      
\usepackage{lipsum}
\usepackage{fancyhdr}       
\usepackage{graphicx}       
\usepackage{todonotes}
\usepackage{comment}
\usepackage{hyperref}

\graphicspath{{media/}}   
\usepackage{url}
\usepackage{tikz}
\usetikzlibrary{lindenmayersystems,calc,decorations,decorations.pathmorphing,decorations.markings,shapes,shapes.geometric,arrows,backgrounds, patterns,decorations.pathreplacing, positioning}
\tikzset{
  currarrow/.style={
    draw=black,
    fill=black,
    single arrow,
    single arrow head extend=0.85mm,
    single arrow tip angle=45,
    minimum height=2.5mm,
    minimum width=1mm,
    inner sep=0pt
  }
}

\tikzset{snake it/.style={decorate, decoration=snake}}
\tikzstyle{vertex} = [circle, draw=black, fill=black, inner sep=0pt,  minimum size=5pt]
\tikzstyle{edgelabel} = [circle, fill=white, inner sep=0pt,  minimum size=15pt]

\newcommand{\TT}{\mathcal{T}}
\newcommand{\XX}{\mathcal{X}}

\definecolor{blue}{RGB}{0,82,147} 
\definecolor{red}{RGB}{202,033,063}

\colorlet{green}{green!50!black}
\definecolor{ForestGreen}{rgb}{.125,.5,.25}

\usepackage{microtype}      
\usepackage{lipsum}
\usepackage{fancyhdr}       
\title{Low-Cost Arborescence Under Edge Faults

}

\usepackage{graphicx} 
\usepackage{todonotes}
\usepackage{comment}
\usepackage{hyperref}

\usepackage{algorithm}
\usepackage{algorithmicx}
\usepackage[noend]{algpseudocode}

\newcommand{\parent}{p}

\newcommand{\opt}{\mathsf{opt}}
\newcommand{\cost}{\mathsf{cost}}
\newcommand{\ftp}{\mathsf{FTP}}
\newcommand{\rank}{\mathsf{rank}}
\newcommand{\Span}{\mathsf{span}}
\newcommand{\SP}{\mathsf{SP}}
\usepackage{tikz}
\usetikzlibrary{decorations.pathmorphing}
\usepackage{subcaption}

\newtheorem{theorem}{Theorem}[section]
\newtheorem{lemma}[theorem]{Lemma}

\theoremstyle{definition}
\newtheorem{definition}[theorem]{Definition}

\newtheorem{proposition}[theorem]{Proposition}
\newtheorem{remark}[theorem]{Remark}
\newtheoremstyle{normalclaim} 
  {3pt}   
  {3pt}   
  {\itshape} 
  {}      
  {\normalfont} 
  {.}     
  { }     
  {}      

\theoremstyle{normalclaim}
\newtheorem{claim}{Claim}

\newtheorem{observation}[theorem]{Observation}

\newenvironment{claimproof}[1][Proof]
  {\begin{proof}[#1]}
  {\end{proof}}

\begin{document}

\begin{center}
    \LARGE
    \textsc{Low-Cost Arborescence Under Edge Faults}
\end{center}
\noindent\rule{\textwidth}{0.4pt}
\vspace{0.1cm}  

\begin{center}
    \begin{minipage}{0.45\textwidth}
        \centering
        \large
        Dipan Dey\\
        \normalsize
        University of Houston\\
        Houston, USA\\
        ddey@cougarnet.uh.edu
    \end{minipage}%
    \hfill
    \begin{minipage}{0.45\textwidth}
        \centering
        \large
        Telikepalli Kavitha\\
        \normalsize
        Tata Institute of Fundamental Research\\
        Mumbai, India\\
        kavitha@tifr.res.in
    \end{minipage}
\end{center}
\vspace{0.1cm}  
\noindent\rule{\textwidth}{0.4pt}

\vspace{0.1cm}
\begin{center}

\textsc{Abstract}
\end{center}
\vspace{0.1cm}
Our input is a directed graph $G = (V,E)$ on $n$ vertices and $m$ edges with a designated root vertex~$r$ and a function 
$\cost: E \rightarrow \mathbb{R}_{\ge 0}$. The problem is to maintain a min-cost arborescence in $G$ in the presence 
of edge faults (a single fault at a time). Edge faults are transient and once the faulty edge is repaired, 
the original min-cost arborescence $\TT$ is restored.
Whenever an edge fault happens, we need to update $\TT$ to a min-cost arborescence
in $G-f$, where $f$ is the faulty edge. 
Since computing a min-cost arborescence in $G - f$ takes $O(m + n\log n)$ time,
we seek to construct a sparse subgraph $H$ in a preprocessing step such that in the event of any edge $f$ failing, it suffices to compute a min-cost arborescence in $H - f$ in order 
to find a low-cost arborescence in $G - f$. 

\smallskip

In the unweighted setting (i.e., $\cost(e) = 1$ for all $e \in E$), this is the fault-tolerant subgraph problem 
for single-source {\em reachability}. Baswana, Choudhary, and Roditty (SICOMP, 2018) showed a $k$-fault tolerant
reachability subgraph of 
size $O(2^kn)$, where $k$ is the number of edge faults; moreover, this bound is tight. To the best of our knowledge, 
this problem has not been addressed in the weighted setting, i.e., when edges have costs. We show a simple polynomial-time
algorithm to construct a subgraph~$H$ of size $O(n^{3/2})$ such that, for any $f \in E$, a min-cost arborescence in 
$H-f$ is a 2-approximation of a min-cost arborescence in $G-f$. Thus whenever an edge fault happens,
we can find a 2-approximate min-cost arborescence in $G-f$ in $O(n^{3/2})$ time.

\smallskip

Our second problem is in the matroid setting. The input is a 
matroid $M = (E, {\cal I})$  with a function $\cost: E \rightarrow \mathbb{R}$.
The problem is to compute a sparse $S \subseteq E$ (called a $k$-fault tolerant preserver) 
such that for any $F \subseteq E$ with $|F| \le k$, the matroid $M|(S\setminus F)$ 
(restriction of $M$ to $S\setminus F$)
contains a min-cost basis of $M|(E\setminus F)$. 
So when $F$ is the set of failed elements, it suffices to compute a min-cost basis in $M|(S \setminus F)$
rather than in $M|(E\setminus F)$. 
We show a tight bound of $k\!\cdot\!\rank(E)$ on the size of a $k$-fault tolerant preserver, where
$\rank: 2^E \rightarrow \mathbb{Z}_{\ge 0}$ is the rank function of~$M$.
\vspace{0.8cm}

\section{Introduction}
\label{sec:intro}
Our input is a directed graph $G = (V, E)$ with a designated root vertex $r$, where
the vertex $r$ has no incoming edge. An {\em arborescence} in $G$ is an acyclic subgraph in which each vertex, except one
called the root, has a unique incoming edge. Since $r$ has no  incoming edge, observe that all arborescences in $G$ have to be rooted at $r$.  Thus every vertex is reachable from $r$ in an arborescence.
There is a function $\cost: E \rightarrow \mathbb{R}_{\ge 0}$ and the problem is to compute a min-cost arborescence in $G$. 
This is a classic problem in graph
algorithms~\cite{Bock71,CL65,Edmonds67} and  the algorithm of Gabow, Galil, Spencer, and Tarjan~\cite{GGST86} finds one in $O(m + n\log n)$ time where $|E| = m$ and $|V| = n$. 

As is well-known with real-world networks, links are prone to failures. It is unlikely that several links fail simultaneously, so we consider the case of a single edge failure. Thus the setting is as follows: we have computed a min-cost 
arborescence $\TT$ in $G$ and an edge $f \in E$ has failed. If $f \in \TT$ then we need to find a new min-cost arborescence $\TT_f$ in $G - f$.\footnote{For convenience, we use $G-f$ to denote the subgraph $(V,E\setminus\{f\})$ of $G$.\label{footnote1}}
We assume any edge fault is transient and there is a repair process that restores the faulty edge.
Thus a min-cost arborescence $\TT_f$ in $G - f$ is an {\em interim} min-cost arborescence---once the edge $f$ is repaired, 
the original min-cost arborescence $\TT$ is restored. However another edge $f'$ might then fail. 
This sequence of edge failure, interim min-cost arborescence, and edge restoration is a recurring step. 

Our goal is to find each of the interim min-cost arborescences efficiently.
Whenever an edge fault $f$ happens, instead of computing a min-cost arborescence $\TT_f$ in $G - f$, 
can we find it more efficiently in real-time if there is a preprocessing step to do some preliminary work? Suppose
the preprocessing step computes for each $f \in \TT$, the min-cost arborescence $\TT_f$ in $G-f$ 
and stores these $n-1$ arborescences as $n-1$ rows in a table. 
Then, whenever there is an edge fault $f \in \TT$, we can fetch the appropriate arborescence $\TT_f$
from this table.
But such a table would take $\Theta(n^2)$ space. We seek a more space-efficient solution.
This motivates the following problem.

\begin{itemize}
    \item Preprocess $G = (V, E)$ and compute a sparse subgraph $H = (V, E')$ 
          such that it suffices to find a min-cost arborescence in $H - f$ whenever an edge $f$ fails.
\end{itemize}

Such a subgraph $H$ is a {\em 1-edge fault tolerant} (1-EFT) subgraph for min-cost arborescence in $G$.
So what we store after the preprocessing step will be the edge set of $H$.
Our objective is to find a sparse 1-EFT subgraph $H$ so as to efficiently process each edge fault $f$ 
by finding a min-cost arborescence in $H - f$ (rather than in $G-f$).

For any directed graph~$G$ with edge costs, does there always exist such a sparse 1-EFT subgraph~$H$? 
Unfortunately, we do not know the answer to this question. 
We show the following result that if we are ready to pay the price of
a small approximation factor, then there is indeed such a sparse subgraph.

\begin{theorem}
    \label{thm:first}
   Let $G = (V, E)$ be any directed graph with $\cost: E \rightarrow \mathbb{R}_{\ge 0}$.
   There is a subgraph $H$ with $O(n^{3/2})$ edges (where $|V| = n$) such that for any $f \in E$,
   the cost of a min-cost arborescence in $H-f$ is at most twice the cost of a min-cost arborescence in $G-f$.
\end{theorem}

We show a simple algorithm that takes $O(mn+n^2\log n)$ time to construct the subgraph~$H$
in a preprocessing step. 
Thus, whenever an edge $f$ fails, we run the algorithm of Gabow et al.~\cite{GGST86} in $H-f$ and find an arborescence of cost at most $2\opt$ in $O(n^{3/2})$ time, where $\opt$ is the cost of a min-cost
arborescence in $G-f$. 
So after each edge fault, our subgraph $H$ provides a 2-approximate solution in $O(n^{3/2})$ time
to our min-cost arborescence problem.
Hence for a dense graph $G$, whenever an edge fault happens,
we achieve significant time-saving at the cost of a 2-approximation.
Furthermore, storing $H$ takes only $O(n^{3/2})$ space.

\subparagraph{The unweighted setting.} The above problem has been well-studied in the unweighted setting, i.e., when $\cost(e) = 1$ for all $e \in E$.
This is the $k$-fault tolerant single-source reachability subgraph ($k$-FTRS) problem, where $k$ is the number of edge/vertex faults.
For any directed graph $G$ on $n$ vertices with a designated source, Baswana, Choudhary, and Roditty~\cite{BCR18} showed
there is a $k$-fault tolerant reachability subgraph of size $O(2^kn)$, where $k$ is the number of edge faults.
They also showed a matching lower bound of $\Omega(2^kn)$ edges for such subgraphs for all $n, k$ where $2^k \le n$.

Prior to their work, the only result for $k$-FTRS was for $k = 1$. The
special case of 1-FTRS is closely related to {\em dominators} by Lengauer and Tarjan~\cite{LT79}. The problem of 1-FTRS was also 
considered in \cite{BCR15} where they showed that given any arbitrary reachability tree $T$, one can efficiently compute
a 1-FTRS with at most $2n$ edges that contains~$T$. 

To the best of our knowledge, this problem has not been addressed in the weighted setting, i.e., when edges have costs.
This is surprising since the weighted variant, which is the problem of maintaining a min-cost arborescence in the presence of edge faults, is very natural. It has applications in telecommunications where data must flow from a central
provider to multiple locations; moreover, different links in the network have different costs. So we seek min-cost reachability from the central provider to all locations and this is the min-cost arborescence problem. 
As is well-known, links may temporarily 
fail. Then, while the faulty links get repaired and restored, in order to maintain reachability from the central 
provider at minimum cost,
we need to update the min-cost arborescence in $G$ to one in $G - F$, where $F$ is the set of faulty edges.

The problem corresponding to min-cost arborescence in the undirected setting is the
min-cost spanning tree (MST) problem. So $G = (V,E)$ is an undirected graph with $\cost: E \rightarrow \mathbb{R}$
and an MST in $G$ has to be maintained in the presence of edge faults. 
This problem has been studied as the  All Best Swap Edges problem for MST~\cite{DixonRT92,Pettie05}.

A more general problem is that of maintaining a min-cost basis in a matroid $M = (E, {\cal I})$
in the presence of faulty elements; we assume there can be up to $k$ failures, for any $k \ge 1$.
The motivation is that elements may get lost or corrupted (at most $k$ at a time) 
and while the lost/corrupted elements are retrieved, we need to update the min-cost basis.

\subparagraph{Fault tolerant preserver for matroid basis.}
Let $M = (E, {\cal I})$ be a matroid  with a function $\cost: E \rightarrow \mathbb{R}$.
For any $S \subseteq E$, the {\em restriction} of matroid $M$ to $S$ is the matroid
$M|S = (S, \,{\cal I}|S)$ where ${\cal I}|S = \{X \in {\cal I}: X \subseteq S\}$. Thus 
$M|S$ is a matroid on $S$ of $\rank(S)$, where $\rank: 2^E \rightarrow \mathbb{Z}_{\ge 0}$ 
is the rank function of $M$.
Our problem is to compute a sparse subset $S$ of $E$, as defined below,
where $k$ is an upper bound on the number of faulty elements.

\begin{definition}
\label{def:k-ftor}
A subset $S \subseteq E$ is called a $k$-fault tolerant
preserver ($k$-$\ftp$) if for every $F \subseteq E$ such that $|F| \le k$, the
restriction $M|(S\setminus F)$ contains a min-cost basis of $M|(E\setminus F)$.
\end{definition}

Our problem is to compute a sparse $k$-$\ftp$, i.e., a sparse subset $S$ of the ground set $E$ such that 
for any subset $F$ of $E$ with $|F| \le k$, there is always a min-cost basis of $M|(E\setminus F)$ that is
present in the matroid $M|(S\setminus F)$. The interpretation is that elements 
in $F$ have failed and a min-cost basis of the surviving matroid, i.e., the matroid $M|(E\setminus F)$, 
needs to be maintained. 
Rather than compute a min-cost basis in $M|(E\setminus F)$, we would like to compute it in 
$M|(S \setminus F)$, for a much smaller set $S$. We show the following result.

\begin{theorem}
    \label{thm:k-ftor}
For any matroid $M = (E, {\cal I})$ with a function $\cost: E \rightarrow \mathbb{R}$ and any integer $k \ge 1$, 
there is a polynomial-time algorithm to compute a $k$-$\ftp$ of size $k\!\cdot\!\rank(E)$ for $M$.
\end{theorem}

\begin{remark}
\label{remark:lb}
    Observe that there is a lower bound of $k\!\cdot\!\rank(E)$ on $k$-$\ftp$. Consider a connected graph $G = (V,E)$ on
    $n$ vertices
where edge costs are distinct. So $G$ has a unique MST $T$. Let $G^* = (V, E^*)$ be the multigraph where each $e \in E$ 
has $k$ copies $e_1,\ldots,e_k$ in $E^*$. Note that a $k$-$\ftp$ for $G^*$ has to contain all $k$ copies of 
$T$ in $G^*$, thus its size is $k(n-1)$.
\end{remark}

\subsection{Related Work}
\label{sec:background}
 A min-cost arborescence rooted at $r$ is very different from a shortest paths tree rooted at $r$ as a shortest paths tree minimizes the distance from $r$ to every other vertex individually, while a min-cost arborescence minimizes the total cost of all edges in the arborescence. Nevertheless,
 the single-source shortest paths problem is perhaps the closest problem to the min-cost arborescence problem.
Similar to the definition of a 1-EFT min-cost arborescence preserver, one can define a 1-EFT
subgraph that preserves the shortest paths tree rooted at $r$ in the presence of single edge faults.
However for weighted graphs (directed or undirected), no sparse 1-EFT subgraph for single-source shortest paths 
is possible. Demetrescu et al.~\cite{DemetrescuTCR08} showed there exist graphs on $n$ vertices whose 1-EFT subgraph for single-source
shortest paths must have $\Omega(n^2)$ edges. For unweighted graphs, Parter and Peleg~\cite{ParterP13} showed
such a 1-EFT subgraph with $O(n^{3/2})$ edges; they also showed a matching lower bound.

For weighted undirected graphs, several results are known for sparse fault-tolerant subgraphs that 
maintain approximate distances from the source vertex $r$ upon the failure of any edge or vertex. 
Baswana and Khanna~\cite{BaswanaK13} showed a subgraph with $O(n\log n)$ edges that, upon the failure of any single vertex, 
approximates distances from $r$ within a multiplicative stretch of 3. Nardelli, Proietti, and Widmayer~\cite{NardelliPW03} showed a subgraph with $2n$ edges that preserves distances from $r$ up to a multiplicative stretch of 3 upon
failure of a single edge. Bil\'o et al.~\cite{BiloGLP14} showed a subgraph with $O(n\log n/{\epsilon}^2)$ edges that preserves a $(1+\epsilon)$-shortest path from the source vertex $r$ after failure of an edge as well as a vertex.

For the unweighted spanning-forest matroid, it was shown in \cite{NI92} that any $k$-edge/vertex connected undirected graph $G$ 
on $n$ vertices has a sparse $k$-connected spanning subgraph with $O(kn)$ edges.
The fault-tolerant matroid basis problem was very recently studied in \cite{BentertFGM25}.
For any matroid $M = (E, {\cal I})$ and a non-negative integer $k$, a set $B \subseteq E$ 
is a $k$-fault-tolerant basis of $M$ if $B$ 
is a min-size set such that $\rank(B\setminus F) = \rank(M)$ for every $F \subseteq B$ with $|F| \le k$.
The problem of deciding if a given matroid admits a $k$-fault-tolerant basis or not 
was considered in \cite{BentertFGM25}, where it was shown that this problem is
NP-hard, even  for $k = 1$. They gave a fixed-parameter tractable (FPT) algorithm for the $k$-fault-tolerant
basis problem, parameterized by both $k$ and the rank $r$ of the matroid.

\subsection{Our Techniques}
\label{sec:techniques}
The $k$-FTRS result for unweighted graphs in \cite{BCR18} shows an algorithm to construct a $k$-FTRS $H$ where the in-degree of every vertex is at most $2^k$: this bounds the size of $H$; but in our setting where edges have costs, there is no guarantee that a 
min-cost arborescence in $H-f$ is a low-cost arborescence in $G-f$.
Moreover, as seen in Section~\ref{sec:background}, sparse fault tolerant subgraphs are rare for problems on weighted directed graphs,
even when we relax the requirement to approximate solutions. 
Thus it is pleasing to see a sparse 1-EFT subgraph for 2-approximate min-cost arborescence. 
Furthermore, our algorithm to construct this sparse 1-EFT subgraph $H$ is truly simple.

The subgraph $H$ consists of a min-cost arborescence $\TT = \{e_1,\ldots,e_{n-1}\}$ in $G$ along with certain paths $\rho_1,\ldots,\rho_{n-1}$. Let vertex $v_i$ be the head of edge $e_i$ (so $e_i$ enters $v_i$).
Each of these paths $\rho_i$ is a shortest path in $G - e_i$ from the component containing the root $r$ in $\TT - e_i$ to $v_i$. 
For any $e_i$, it is easy to show that a min-cost arborescence in $H - e_i$ is a 2-approximate min-cost
arborescence in $G - e_i$. 

We use a simple and clean charging method to bound the size of $H$. 
The edges of each path $\rho_i$ are charged to pairs 
of vertices
$(x,y)$ such that (i)~$\rho_i$ has a subpath from $x$ to $y$, 
(ii)~$x$ gets a new outgoing edge from~$\rho_i$ and (iii)~$y$ gets a new incoming edge from $\rho_i$.
Any particular path $\rho_i$ may charge $\Theta(n^2)$ pairs of vertices.
Nevertheless, the total charge paid by all the $n^2$ pairs of vertices summed over the $n-1$ paths $\rho_1,\ldots,\rho_{n-1}$
is only $3n^2$
(see Lemma~\ref{lem_two_charged_paths}). 
This bound of $3n^2$, along with the Cauchy-Schwarz inequality, allows us to bound the number of edges in $H$ by $O(n^{3/2})$.

We would like to point out that this type of analysis has been seen in the literature on distance preservers to
bound the number of edges in a collection of (non-faulty) shortest paths~\cite{Bodwin,CE05}, e.g., the union of
any $p$ shortest paths in a directed weighted graph is bounded by $O(n\sqrt{p})$ edges~\cite{CE05}. However our result is technically stronger
because it deals with {\em 1-edge fault replacement paths} while previous work dealt with exact shortest paths. Indeed, the
graph on which we compute each replacement path (see step~4 of Algorithm~\ref{alg0}) 
is changing as per the faulty edge while the graph remains static in the setting of  exact shortest paths.
Section~\ref{sec:approx} has all the details. 
In the appendix we discuss issues involved in extending our construction to a 1-EFT preserver for min-cost arborescence.

Our $k$-$\ftp$ for min-cost basis of matroid is based on the idea of ``swap elements''. 
Let $B_1$ be a min-cost basis of $M$. It is easy to show that a min-cost basis $B_2$ of $M|(E \setminus B_1)$ has an element $e$ for each $f \in B_1$ such that $B_1 - f + e$ is a min-cost basis of $M|(E \setminus \{f\})$. Extending this idea for $k+1$ levels yields our $k$-$\ftp$. These details are in Section~\ref{sec:matroids}.
We discuss preliminaries in Section~\ref{sec:prelims} and conclude in Section~\ref{sec:discussion}.

\section{Preliminaries}
\label{sec:prelims}
This section describes notation that will be used in the rest of the paper.
Our input is a directed graph $G = (V,E)$ with a designated root vertex $r$
and a function $\cost: E \rightarrow \mathbb{R}_{\ge 0}$. 
An arborescence is a directed tree rooted at $r$.

For paths, it will be convenient to view the $\cost$ function on the edge set as a length function. Thus
a {\em shortest path} from vertex $u$ to vertex $v$ is a path of minimum cost from $u$ to $v$. 
For any path $\rho$, let $||\rho||$ be the sum of costs of edges in $\rho$, i.e., $||\rho|| = \sum_{e\in\rho} \cost(e)$.

\begin{itemize}
    \item  For any pair of vertices $(u,v)$, it will be convenient to assume that the shortest path from $u$ to $v$
is unique. A random perturbation of the given edge costs achieves the property of unique shortest paths: 
see \cite{ParterP13}. The  property of unique shortest paths has been used in several previous works, e.g., \cite{BernsteinK09,DuanR22,GuptaS18,HershbergerS01}. 
\end{itemize}
We denote the shortest path from $u$ to $v$ by $\SP(u,v)$. 
So $||\SP(u,v)||$ is the distance from $u$ to $v$ in~$G$, i.e., the sum of costs of edges on the shortest path from $u$ to $v$.

\medskip

Let $G-F = (V, E\setminus F)$ be the graph obtained after deleting all edges in $F$ from the graph $G$. 
As in $G$, we assume that in $G-F$ as well, there is a unique shortest path from one vertex to another.
For any $(u,v) \in V \times V$, let $\SP(u,v) \diamond F$ be the shortest path from $u$ to $v$ in $G - F$. 
So $||\SP(u,v)\diamond F||$ is the distance from $u$ to $v$ in $G-F$. 
When $F = \{f\}$, we will denote $G - F$ by $G - f$ (see Footnote~\ref{footnote1}). 
We will use $||\SP(u,v)\diamond f||$ to denote the distance from $u$ to $v$ 
in $G-f$. Similarly, we will use $\TT - f$ to denote $\TT\setminus\{f\}$.

Let $G^r = (V,E^r)$ denote the {\em reverse graph} of $G$, i.e., every edge $(a,b) \in E$ is replaced by $(b,a)$ in $E^r$. Thus edge directions in $G$ are reversed in $G^r$; also
$\cost(b,a)$ in $G^r$ will be equal to $\cost(a,b)$ in $G$, for any $(a,b) \in E$.

\subparagraph{Matroids.} We refer to the book by Oxley~\cite{Oxley92} and lecture notes by Goemans~\cite{Goemans}
for an introduction to matroids.
A matroid $M$ is defined on a finite ground set $E$ and a collection ${\cal I} \subseteq 2^E$: the sets in ${\cal I}$ 
are said to be independent. For $M$ to be a matroid, the collection ${\cal I}$ should be {\em downward-closed} 
($X \in {\cal I}$ implies all subsets of $X$ are also in ${\cal I}$) and 
the {\em exchange} property has to hold in ${\cal I}$. That is, if $X, Y \in {\cal I}$ with $|Y| > |X|$ then $\exists\, e \in Y \setminus X$
such that $X\cup\{e\} \in {\cal I}$. We will use $X + e$ to denote $X \cup \{e\}$ and $X - f$ to denote $X \setminus \{f\}$.

\begin{itemize}
    \item An inclusion-wise maximal independent set $B$ is called a {\em basis} of $M$. All bases of $M$ 
    have the same size called the {\em rank} of $M$.
    \item The rank of a subset $S \subseteq E$, denoted by $\rank(S)$, is the maximum
size of an independent set $X \subseteq S$; the function $\rank : 2^E \rightarrow \mathbb{Z}_{\ge 0}$ 
is the rank function of $M$.
   \item A set $S \subseteq E$ spans an element $x \in E$ if $\rank(S + x) = \rank(S)$. The {\em span} of $S$ is the set 
   $\Span(S) = \{x \in E: S\ \text{spans}\ x\}$.
\end{itemize}

\section{Constructing a 1-EFT Subgraph for Low-Cost Arborescence}
\label{sec:approx}

Let $\TT$ be a min-cost arborescence in $G = (V,E)$ with $\cost: E \rightarrow \mathbb{R}_{\ge 0}$. 
As mentioned earlier, there are several polynomial-time algorithms~\cite{Bock71,CL65,Edmonds67,GGST86} 
known to compute a min-cost arborescence 
in $G$ and the fastest among these takes $O(m + n\log n)$ time~\cite{GGST86}, where $|E| = m$ and $|V| = n$.

For any $v \in V \setminus \{r\}$, let $\parent(v)$ be $v$'s parent in the arborescence $\TT$. 
Let $f = (p(v),v) \in \TT$.
Deleting the edge $f$ from the input graph $G$ will partition the arborescence $\TT$ in two parts: 
\begin{itemize}
    \item subtree $\TT(v)$: this is the subtree of $\TT$ rooted at $v$ and
    \item subtree  $\XX(v)$: this is the component containing $r$ in $\TT - f$.
\end{itemize}

So $\XX(v)$ is the  subtree $(\TT - f)\setminus\TT(v)$.
There is no path in $\TT - f$ from $r \in \XX(v)$ to $\TT(v)$ (see Figure~\ref{fig1}). 
Let $\rho_v$ denote the shortest path from $\XX(v)$ to $v$ in $G - f$, i.e., $\rho_v$ is the shortest path in $G - f$ from any vertex in 
$\XX(v)$ to $v$, where each edge $e$ has weight $\cost(e)$. 

   \begin{figure}[hpt!]
\centering

        \begin{tikzpicture}[scale=0.5]   
        \coordinate (r) at (5,15);
        \coordinate (b) at (2.5,8);
        \coordinate (a) at (3,10);
        \coordinate (g3) at (6,5);
        
       \draw[very thick, decorate, decoration={snake, amplitude=0.5mm, segment length=5mm}] (r) -- (11,2);
       \draw[very thick, decorate, decoration={snake, amplitude=0.5mm, segment length=5mm}] (r) -- (a);
       \draw[very thick, decorate, decoration={snake, amplitude=0.5mm, segment length=5mm}] (b) -- (1,2);
       \draw[very thick, decorate, decoration={snake, amplitude=0.5mm, segment length=5mm}] (b) -- (6,2);
       \draw[very thick, decorate, decoration={snake, amplitude=0.5mm, segment length=5mm}] (a) -- (8,2);  
       \draw[very thick, red, ->] (a) -- (b); 
       \draw[very thick] (1,2) -- (6,2);
       \draw[very thick] (8,2) -- (11,2); 
       \node[left] at (a) {$p(v)$};
       \node[left] at (b) {$v$};
       \node[above] at (8,5) {$u$};
       \node[above] at (6.4,8) {$\XX(v)$};
       \node[above] at (3,3) {$\TT(v)$};
       \node[above] at (r) {$r$};
       \path [very thick,draw,cyan, ->] (8,5) .. controls (5,2) and (4,2) ..  (2.5,7.7);  
         \end{tikzpicture}
        \caption{Let $f = (p(v),v)$. A shortest path $\rho_v$ from $\XX(v)$ to $v$ in $G - f$ is highlighted in \textcolor{cyan}{cyan}.}
         \label{fig1}
\end{figure}
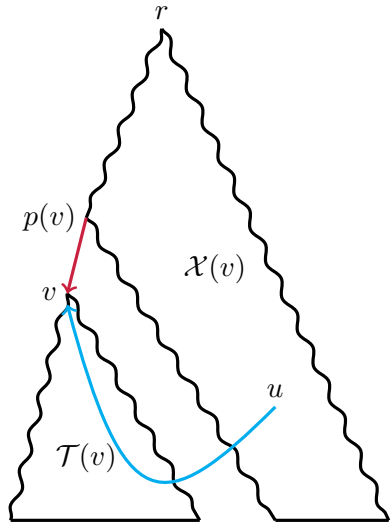

\begin{algorithm}
\caption{Computing a sparse 1-EFT subgraph $H$ for a low-cost arborescence}
\label{alg0}
\begin{algorithmic}[1]
\State Label the vertices of $V \setminus \{r\}$ as $v_1,\ldots,v_{n-1}$ arbitrarily.
\State Compute a min-cost arborescence $\TT$ in $G = (V,E)$.
\For{each edge $f = (\parent(v_i),v_i) \in \TT$}
\State Find a shortest path $\rho_i$ from $\XX(v_i)$ to $v_i$ in $G - f$.
\EndFor
\State Return $H = (V,E_H)$ where $E_H = \TT \, \bigcup_{i=1}^{n-1}\rho_i$.
\end{algorithmic}
\end{algorithm}   

Our algorithm is described as Algorithm~\ref{alg0}.    
A shortest path $\rho_i$ from $\XX(v_i)$ to $v_i$ (in step~4) 
can be obtained by running Dijkstra's algorithm from $v_i$ in the reverse graph $(G - f)^r$; the nearest vertex in $\XX(v_i)$ 
to $v_i$ in $(G - f)^r$ will be the starting vertex $u$ on $\rho_i$ (see Figure~\ref{fig1}). 
Thus Algorithm~\ref{alg0} runs $n-1$ shortest path computations in
various subgraphs of $G^r$. Hence its running time is $O((m + n\log n)\cdot n) = O(mn + n^2\log n)$.

\subparagraph{Edge faults.}
Suppose an edge fault $f = (\parent(v_i),v_i)$ occurs. If $G-f$ admits an arborescence then $v_i$ is reachable from
$\XX(v_i)$ in $G-f$, so $\rho_i$ (see step~4 of Algorithm~\ref{alg0}) exists. Thus $H - f$ also 
admits an arborescence since the path $\rho_i$ and all edges of $\TT - f$ are present in $H-f$. 

Let $\TT_f$ be a min-cost arborescence in $G - f$ and let $A_f$ be a min-cost arborescence in $H - f$.
The following lemma is easy to show.

\begin{lemma}
\label{lem : preserver one vertex}
   For $A_f$ and $\TT_f$ as defined above, $\cost(\TT_f) \le \cost(A_f) \leq 2\cost(\TT_f)$.
\end{lemma}
 \begin{proof}
    We have $\cost(\TT_f) \le \cost(A_f)$ because $\TT_f$ is a min-cost arborescence in $G - f$. 
    Since $A_f$ is a min-cost arborescence in $H-f$ and all edges of 
    $(\TT-f) \cup \rho_i$ are in $H-f$, we have 
    $\cost(A_f) \le \cost(\TT) + ||\rho_i||$.
    Recall that $\rho_i$ is a shortest path in $G - f$ from $\XX(v_i)$ to $v_i$ where $r \in \XX(v_i)$. So $||\rho_i|| \le ||\SP(r,v_i)\diamond f||$,
    where $||\SP(r,v_i)\diamond f||$ is the distance from $r$ to $v_i$ in $G - f$.

    Observe that $\cost(\TT_f) \ge ||\SP(r,v_i)\diamond f||$ since $v_i$ is reachable from $r$ in
    $\TT_f$. So $\TT_f$ has to contain a path from $r$ to $v_i$ in $G-f$.
    Hence $||\rho_i|| \le ||\SP(r,v_i)\diamond f|| \le \cost(\TT_f)$. So it follows that 
    $\cost(A_f) \ \le \ \cost(\TT) + ||\rho_i|| \ \le \ \cost(\TT) + \cost(\TT_f) \ \le \ 2\cost(\TT_f)$.
\end{proof}

Proposition~\ref{prop1} follows immediately from Figure~\ref{lem : preserver one vertex}.
\begin{proposition}
    \label{prop1}
The subgraph $H$ is a 1-EFT 2-approximate min-cost arborescence preserver.  
\end{proposition}

In order to conclude Theorem~\ref{thm:first}, we need to bound $|E_H|$ by $O(n^{3/2})$. Though any particular path $\rho_i$ 
may have $\Theta(n)$ edges, we will show in Section~\ref{sec:analysis} that the union of all these $n-1$ paths has only $O(n^{3/2})$ edges.

\subsection{A Charging Method}
\label{sec:analysis}
In building the subgraph $H$, assume without loss of generality that 
Algorithm~\ref{alg0} adds edges in the order $\rho_1, \rho_2, \ldots,\rho_{n-1}$. Thus
edges of path $\rho_1$ are added first, then edges of $\rho_2$, so on, and finally edges of $\rho_{n-1}$.
We will use a coloring scheme as follows. 

For $i = 1$ to $n-1$:
\begin{itemize}
    \item Edges belonging to $\rho_i \setminus (\rho_1\cup\cdots\rho_{i-1})$ are colored $i$.
\end{itemize}

Thus every edge $e \in \rho_1\cup\cdots\cup\rho_{n-1}$ gets a color, which is the index of the least-indexed path 
in $\{\rho_1,\ldots,\rho_{n-1}\}$ that contains $e$.
Observe that all color~1 edges form a contiguous path while this property need not hold for higher colored edges. That is, for any $i > 1$, color~$i$ edges may be scattered since some edges in the path $\rho_i$ were already colored by lower colors.

\begin{definition}
    Let $(x,y) \in V \times V$ where $x \ne y$. For any $i \in \{1,\ldots,n-1\}$, we will say {\em color} $i$ is \emph{charged to $(x,y)$} if the following three conditions are satisfied:
    \begin{enumerate}
        \item $\rho_i$ contains a subpath from $x$ to $y$;
        \item $x$ has an outgoing edge colored $i$;
        \item $y$ has an incoming edge colored $i$.
    \end{enumerate}
\end{definition}

Condition~2 says there is some edge $(x,a)$, i.e., edge leaving $x$, that is colored $i$ 
while condition~3 says there is some edge $(b,y)$, i.e., edge entering $y$, that is colored $i$. 
We will prove the following lemma in Section~\ref{sec:lemma8}.

\begin{lemma}
\label{lem_two_charged_paths}
For any pair $(x,y) \in V \times V$, at most three colors are charged to $(x,y)$.
\end{lemma}

For now, we will assume Lemma~\ref{lem_two_charged_paths} and bound the number of edges in the subgraph $H$ in the following lemma.

\begin{lemma}
\label{lem_upper_bound}
    The number of edges in the set $\cup_{i=1}^{n-1}\rho_i$ is at most $\sqrt{6}\,n^{3/2}$.
\end{lemma}
\begin{proof}
    We have to bound $\sum_{i=1}^{n-1} c_i$, where $c_i$ is the number of edges colored $i$. It follows from the Cauchy-Schwarz inequality that:
    \begin{eqnarray}
    \sum_{i=1}^{n-1} c_i \ \ = \ \ c_1 \cdot 1+c_2 \cdot 1+\dots+c_{n-1} \cdot 1 \ \ & \le & \ \ \sqrt{1+1+\dots +1} \,\cdot\, \sqrt{c_1^2+c_2^2 +\dots+c_{n-1}^2}\nonumber\\
           & = & \sqrt{n-1}\,\cdot\, \sqrt{\sum_{i=1}^{n-1}c_i^2}.\label{cauchy-schwarz}
    \end{eqnarray}
    
    Since $c_i$ is  the number of edges colored $i$, there are $c_i$ vertices with an outgoing edge colored $i$ and
    there are $c_i$ vertices with an incoming edge colored $i$.\footnote{The same vertex can have an incoming edge colored $i$ and also an outgoing edge colored $i$.} So the number of pairs $(x,y)$ that get charged by color~$i$ is $(c_i + (c_i-1) + (c_i-2) + \ldots + 1) = c_i(c_i+1)/2$ where the term $(c_i - k + 1)$ in this sum 
    is contributed by the $k$-th vertex with an outgoing edge colored $i$ on $\rho_i$, for $1 \le k \le c_i$
    (see Figure~\ref{fig3}). 
    
    \medskip
    
    \begin{figure}[h]
\centering
\begin{tikzpicture}[scale=1, transform shape]
       \coordinate (x) at (-5,0);
        \coordinate (y) at (5,0);
         \coordinate (a) at (-4,0);
        \coordinate (b) at (-3,0);
        \coordinate (c) at (4,0);
         \coordinate (x') at (0,0);
        \coordinate (y') at (1,0);
        
        \filldraw [black] (x) circle(2pt);
         \filldraw [black] (y) circle(2pt);
        \filldraw [black] (c) circle(2pt);
         \filldraw [black] (a) circle(2pt);
         \filldraw [black] (b) circle(2pt);
        \filldraw [black] (x') circle(2pt);
         \filldraw [black] (y') circle(2pt);
          \filldraw [black] (-2,0) circle(2pt);
           \filldraw [black] (-1,0) circle(2pt);
            \filldraw [black] (2,0) circle(2pt);
             \filldraw [black] (3,0) circle(2pt);

         \node[below, scale=1] at (x) {$x_1$};
         \node[below, scale=1] at (y) {$x_{10}$};
         \node[below, scale=1] at (a) {$x_2$};
         \node[below, scale=1] at (b) {$x_3$};
         \node[below, scale=1] at (0,0) {$x_6$};
         \node[below, scale=1] at (1,0) {$x_7$};
         \node[below, scale=1] at (4,0) {$x_9$};
        
\draw[very thick, red, ->] (-5,0) -- (-4.1,0);
\draw[very thick, red, ->] (-4,0) -- (-3.1,0);
\draw[thin, ->] (-3,0) -- (-2.1,0);
\draw[thin, ->] (-2,0) -- (-1.1,0);
\draw[thin, ->] (-1,0) -- (-0.1,0);

\draw[very thick, red, ->] (0,0) -- (0.9,0);

\draw[thin, ->] (1,0) -- (1.9,0);
\draw[thin, ->] (2,0) -- (2.9,0);
\draw[thin, ->] (3,0) -- (3.9,0);
\draw[very thick, red, ->] (4,0) -- (4.9,0);
\end{tikzpicture}
\caption{The path $\rho_i$ is from $x_1$ to $x_{10}$ and $c_i = 4$ (corresponding to the 4 red edges) because the remaining edges in $\rho_i$ were already present in $\rho_1\cup\cdots\cup\rho_{i-1}$. The 10 pairs $(x_1,x_2),(x_1,x_3)$,
 $(x_1,x_7),(x_1,x_{10}),(x_2,x_3),(x_2,x_7),(x_2,x_{10}),(x_6,x_7),(x_6,x_{10})$, and $(x_9,x_{10})$ are charged by color $i$.}
\label{fig3}
\end{figure}
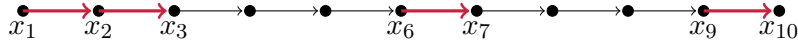
    
    We have:
    \begin{eqnarray*}
        \sqrt{\sum_{i=1}^{n-1}c_i^2} \ \ & \le & \ \ \left(\sum_{i=1}^{n-1}c_i(c_i+1)\right)^{1/2}\\
                                                          \ \ & = & \ \ \left(2\cdot\sum_i\text{number of pairs of vertices charged by color}\ i\right)^{1/2}\\
                                                          \ \ & = & \ \ \left(2\cdot\sum_{(x,y)\in V\times V}\text{number of colors charged to}\ (x,y)\right)^{1/2}\\
                                                          \ \ & \le & \ \ \left(2\cdot 3n^2\right)^{1/2} = \sqrt{6}\,n,
     \end{eqnarray*}
     where the last inequality 
     follows from the fact that any pair of vertices gets charged by at most three colors
     (see Lemma~\ref{lem_two_charged_paths}). 
     Hence the right hand side of Equation~\eqref{cauchy-schwarz} can be bounded from above by $\sqrt{6}n^{3/2}$.
     So the number of edges in $\cup_{i=1}^{n-1}\rho_i$ is at most $\sqrt{6}n^{3/2}$. 
\end{proof}

Thus, assuming Lemma~\ref{lem_two_charged_paths}, the number of edges in the subgraph $H$ is 
$O(n^{3/2})$.
What is left to show is  Lemma~\ref{lem_two_charged_paths}.
For any pair $(x,y) \in V \times V$, recall that $\SP(x,y)$ is the shortest path from $x$ to $y$ in $G$.
The following observation will be useful.

\begin{observation}
\label{obs_disjoint}
    Let $(x,y) \in V \times V$. Suppose the path $q_i = \SP(x,y) \diamond e_i$ and the path $q_j = \SP(x,y) \diamond e_j$ for some edges $e_i$ and $e_j$.
    If $q_i \ne q_j$ then $e_i \in q_j$ or $e_j \in q_i$.
\end{observation}
\begin{proof}
    If neither $e_i \in q_j$ nor $e_j \in q_i$ then both $q_i$ and $q_j$ are the same as $\SP(x,y) \diamond \{e_i,e_j\}$, 
    which is the shortest path from $x$ to $y$ in $G - \{e_i,e_j\}$. By the uniqueness of shortest paths in $G - F$ 
    for any $F \subseteq E$, it follows that $q_i = q_j$.
\end{proof}

For any path $\rho_i$ (see step~4 of Algorithm~\ref{alg0}) 
that contains vertices $x$ and $y$, let $\rho_i[x,y]$ denote the subpath from $x$ to $y$ in $\rho_i$.
Lemma~\ref{cor1} will be helpful in proving Lemma~\ref{lem_two_charged_paths}
and Observation~\ref{obs_disjoint} is key to proving Lemma~\ref{cor1}.

\begin{lemma}
    \label{cor1}
    Let $i$ and $j$ be colors such that both $\rho_i[x,y]$ and $\rho_j[x,y]$ are edge-disjoint from $\SP(x,y)$.
    Then $\rho_i[x,y] = \rho_j[x,y]$.
\end{lemma}
\begin{proof}
    The path $\rho_i$ is a shortest path in $G - e_i$ from some vertex $u_i$ to $v_i$, for some edge~$e_i$. 
    Because a subpath of a shortest path is again a shortest path, we have
    $\rho_i[x,y] = \SP(x,y) \diamond e_i$.
    Since $\rho_i[x,y]= \SP(x,y) \diamond e_i$ is edge-disjoint from the shortest path $\SP(x,y)$ in $G$, 
    it follows that $e_i \in \SP(x,y)$; 
    otherwise (so $e_i \notin \SP(x,y)$) it would have been the case that $\rho_i[x,y] = \SP(x,y)$. 
    Similarly, $\rho_j[x,y] = \SP(x,y) \diamond e_j$ for some $e_j \in \SP(x,y)$. 
    
    Since each of $\rho_i[x,y]$ and $\rho_j[x,y]$ is edge-disjoint from $\SP(x,y)$ while both $e_i$ and $e_j$ belong to $\SP(x,y)$, 
    it follows that neither $e_i$ nor $e_j$ is present in $\rho_i[x,y] \cup \rho_j[x,y]$. 
    Hence $\rho_i[x,y] = \rho_j[x,y]$ (by Observation~\ref{obs_disjoint}).
\end{proof}

\subsection{Bounding the Total Charge Paid by a Pair of Vertices}
\label{sec:lemma8}
We will prove Lemma~\ref{lem_two_charged_paths} in this section.
Consider any index $i \in \{1,\ldots,n-1\}$ and any pair $(x,y) \in V \times V$ such that
$\rho_i$ contains a subpath from $x$ to $y$. The following definition will be useful.

\begin{definition}
    \label{def:overlapping}
    Call index $i$ ``intersecting for $(x,y)$'' if $\rho_i[x,y]$ is not edge-disjoint from $\SP(x,y)$, i.e.,
$\rho_i[x,y] \cap \SP(x,y) \ne \emptyset$. 
\end{definition}

Call index $i$ {\em non-intersecting} (i.e., disjoint) for $(x,y)$ if $\rho_i[x,y] \cap \SP(x,y) = \emptyset$, in other words,
if $\rho_i[x,y]$ is edge-disjoint from $\SP(x,y)$ (see Figure~\ref{fig2}). 
We are now ready to prove Lemma~\ref{lem_two_charged_paths}, i.e., for any $(x,y) \in V \times V$, at most three colors 
are charged to $(x,y)$. 

\begin{figure}[h]

\centering
\begin{tikzpicture}[scale=1, transform shape]
      
       \coordinate (x) at (-5,0);
        \coordinate (y) at (5,0);
         \coordinate (a) at (-4,0);
        \coordinate (b) at (4,0);
        \coordinate (z) at (3,0);
        \coordinate (a') at (-3.1,0);
        \coordinate (a'') at (-2.2,0);
         \coordinate (b') at (0,0);
        \coordinate (b'') at (1,0);
        
         \coordinate (x') at (-1.5,0);
        \coordinate (y') at (2,0);
        
        \filldraw [black] (x) circle(3pt);
         \filldraw [black] (y) circle(3pt);

        \filldraw [black] (b') circle(2pt);
        \filldraw [black] (b'') circle(2pt);
        \filldraw [black] (a') circle(2pt);
        \filldraw [black] (a'') circle(2pt);
         \filldraw [black] (a) circle(2pt);
         \filldraw [black] (b) circle(2pt);
        \filldraw [black] (x') circle(2pt);
         \filldraw [black] (y') circle(2pt);
        \filldraw [black] (z) circle(2pt);

         \draw [very thick, cyan] (-5,0) -- (2,0.04);
         \draw [very thick, magenta] (5,0) -- (-1.5,0);
         

         \node[above, scale = 1,magenta] at (-2.7,0) {$e_1$};
         \node[above, scale = 1,cyan] at (3.5,0) {$e_3$};
         \node[above, scale = 1,orange] at (0.5,0) {$e_2$};
         
         \node[left, scale=1] at (x) {$x$};
         \node[right, scale=1] at (y) {$y$};
         \node[below, scale=1] at (a) {$a$};
         \node[below, scale=1] at (b) {$b$};
         \node[below, scale=1] at (-1.5,0) {$x'$};
         \node[below, scale=1] at (2,0) {$y'$};
         
           
\draw[very thick, cyan, ->] (-5,0) -- (-4.1,0); 
\draw[very thick, magenta, ->] (4,0) -- (4.9,0); 

\draw[very thick, cyan] (-1.5,0.05) -- (2,0.05);
           
           \path [very thick,draw,magenta] (x) .. controls (-4.5,1.5) and (-2,1.5) ..  (-1.5,0);
           \path [very thick,draw,cyan] (2,0) .. controls (2.5,1.5) and (4.5,1.5) ..  (y);
           \path [very thick,draw,orange] (x) .. controls (-3.5,-1.5) and (3.5,-1.5) ..  (y);


\end{tikzpicture}
\caption{The horizontal path from $x$ to $y$ is $\SP(x,y)$.
The {\color{magenta} magenta} path $\rho_1$ is $\SP(x,y) \diamond e_1$ for some $e_1 \in \SP(x,x')$ while
the {\color{orange} orange} path $\rho_2$ is $\SP(x,y) \diamond e_2$ for some $e_2 \in \SP(x',y')$ and 
the {\color{cyan} cyan} path $\rho_3$ is $\SP(x,y) \diamond e_3$ for some $e_3 \in \SP(y',y)$.
For the pair $(x,y)$, indices 1 and 3 are intersecting while index 2 is non-intersecting.}
\label{fig2}
\end{figure}
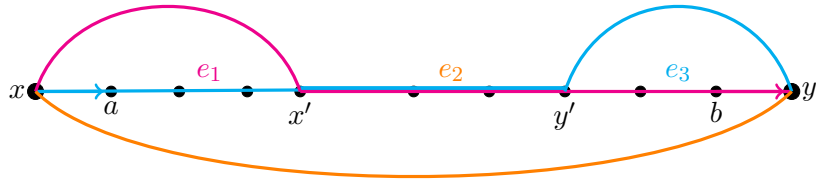

\begin{proof}[Proof of Lemma~\ref{lem_two_charged_paths}.]
Our first claim is that 
there is at most one non-intersecting index $k$ for which the pair $(x,y)$ gets charged. This is because for any two paths $\rho_k, \rho_{k'}$ with subpaths from $x$ to $y$ that are edge-disjoint from $\SP(x,y)$, we have $\rho_k[x,y] = \rho_{k'}[x,y]$ (by Lemma~\ref{cor1}). Thus the pair $(x,y)$ gets charged by at most one non-intersecting index $k \in \{1,\ldots,n-1\}$.

We now have to show that there are at most two intersecting indices $\ell \in \{1,\ldots,n-1\}$ such that 
(i)~$\rho_\ell$ contains a subpath from $x$ to $y$,
(ii)~$x$ has an outgoing edge colored $\ell$, and (iii)~$y$ has an incoming edge colored~$\ell$. 
Let $(x,a)$ and $(b,y)$ be the first and last edges on $\SP(x,y)$ (see Figure~\ref{fig2}). 
We will show the following claim.

\medskip

\begin{claim}
\label{clm1}
If $\ell$ is an intersecting index for the pair $(x,y)$ then (i) $(x,a) \in \rho_\ell[x,y]$ or (ii)~$(b,y) \in \rho_\ell[x,y]$.     
\end{claim}   
\begin{claimproof}
Suppose neither $(x,a)$ nor $(b,y)$ belongs to $\rho_\ell[x,y]$.
This means $\rho_\ell[x,y] \ne \SP(x,y)$. 
Hence $\rho_\ell[x,y] = \SP(x,y) \diamond e_\ell$ where $e_\ell \in \SP(x,y)$.
Since $\ell$ is an intersecting index, there is some interior vertex $z$ on the path $\SP(x,y)$ 
that lies on $\rho_\ell[x,y]$ (see Figure~\ref{fig30}).

\begin{figure}[h]

\centering
\begin{tikzpicture}[scale=1, transform shape]
      
       \coordinate (x) at (-5,0);
        \coordinate (y) at (5,0);
         \coordinate (a) at (-4,0);
        \coordinate (b) at (4,0);

         \coordinate (z) at (0,0);
        
        \filldraw [black] (z) circle(2pt);

         \filldraw [black] (a) circle(2pt);
         \filldraw [black] (b) circle(2pt);
        \filldraw [black] (x) circle(2pt);
         \filldraw [black] (y) circle(2pt);

         \draw [very thick, magenta] (-1.5,0.05) -- (2,0.05);
         \draw [very thick, black] (-5,0) -- (5,0);
         

         
         \node[left, scale=1] at (x) {$x$};
         \node[right, scale=1] at (y) {$y$};
         \node[below, scale=1] at (a) {$a$};
         \node[below, scale=1] at (b) {$b$};
         \node[below, scale=1] at (0,0) {$z$};
         
           
\draw[very thick, ->] (-5,0) -- (-4.1,0); 
\draw[very thick, ->] (4,0) -- (4.9,0); 

           
           \path [very thick,draw,magenta] (x) .. controls (-4.5,1.5) and (-2,1.5) ..  (-1.5,0);
           \path [very thick,draw,magenta] (2,0) .. controls (2.5,1.5) and (4.5,1.5) ..  (y);


\end{tikzpicture}
\caption{The horizontal path from $x$ to $y$ is $\SP(x,y)$. The {\color{magenta} magenta} path $\rho_{\ell}[x,y] = \SP(x,y) \diamond e_\ell$ for some $e_\ell \in \SP(x,y)$. We assumed neither $(x,a)$ nor $(b,y)$ belongs to $\rho_{\ell}[x,y]$.}
\label{fig30}
\end{figure}
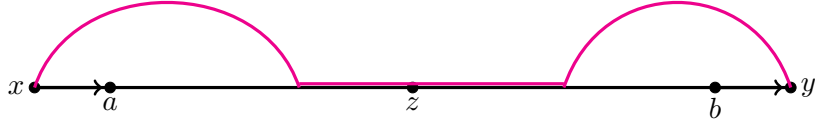

Consider the following two cases.

\begin{enumerate}
    \item $e_\ell \in \SP(x,z)$: Since $e_\ell \in \SP(x,z)$, the entire path $\SP(z,y)$ belongs to $G - e_\ell$. 
    Observe that $\SP(z,y) \ne  \rho_\ell[z,y]$, because the edge $(b,y) \in \SP(z,y)$ does not belong to  $\rho_\ell[z,y]$.
    Moreover, $\rho_\ell[z,y] > ||\SP(z,y)||$ due to unique shortest paths.
    So the $x$ to $y$ path obtained 
    by stitching the paths $\rho_\ell[x,z]$ and  $\SP(z,y)$ is a shorter $x$ to $y$ path in $G - e_\ell$ 
    than  $\rho_\ell[x,y] = \SP(x,y) \diamond e_\ell$, a contradiction.
    
    \item $e_\ell \in \SP(z,y)$: Since $e_\ell \in \SP(z,y)$, the entire path $\SP(x,z)$ belongs to $G - e_\ell$. 
    Observe that $\SP(x,z) \ne  \rho_\ell[x,z]$, because the edge $(x,a) \in \SP(x,z)$ does not belong to  $\rho_\ell[x,z]$.
    Moreover, $\rho_\ell[x,z] > ||\SP(x,z)||$ due to unique shortest paths.
    So the $x$ to $y$ path obtained by stitching the paths $\SP(x,z)$ and $\rho_\ell[z,y]$ 
    is a shorter $x$ to $y$ path in $G - e_\ell$ 
    than  $\rho_\ell[x,y] = \SP(x,y) \diamond e_\ell$, a contradiction..
\end{enumerate}

Thus we can conclude that for any intersecting index $\ell$, we have 
$(x,a) \in \rho_\ell[x,y]$ or $(b,y) \in \rho_\ell[x,y]$ (or possibly both).
\end{claimproof}

Once the edge $(x,a)$ gets colored by the least index $i \in \{1,\ldots,n-1\}$ such that $(x,a) \in \rho_i$
and the edge $(y,b)$ gets colored by the least index $j \in \{1,\ldots,n-1\}$ such that $(b,y) \in \rho_j$,
no other intersecting index $\ell \in \{1,\ldots,n-1\}$ can satisfy both the following conditions:
\begin{itemize}
        \item $x$ has an outgoing edge colored $\ell$;
        \item $y$ has an incoming edge colored $\ell$.
\end{itemize}

This is because at least one of the edges $(x,a),(b,y)$ is in $\rho_\ell[x,y]$ (by Claim~\ref{clm1}) 
and both these edges have already been colored $i$ and $j$, respectively.
So any pair $(x,y)$ gets charged by at most two intersecting indices (or colors) $i$ and $j$. We already showed that
$(x,y)$ gets charged by at most one non-intersecting index (or color) $k$. Hence the lemma follows.
\end{proof}

This finishes the proof of Theorem~\ref{thm:first} stated in Section~\ref{sec:intro}.

\section{Fault tolerant preserver for min-cost matroid basis}
\label{sec:matroids}
Let $M = (E, {\cal I})$ be a matroid  with a cost function $c: E \rightarrow \mathbb{R}$.
Let $\rank: 2^E \rightarrow \mathbb{Z}_{\ge 0}$ be the rank function of $M$.
For any $S \subseteq E$, the {\em restriction} of matroid $M$ to $S$ is the matroid
$M|S = (S, {\cal I}|S)$ where ${\cal I}|S = \{X \in {\cal I}: X \subseteq S\}$. Thus 
$M|S$ is a matroid on $S$ of $\rank(S)$.

We consider the problem of computing a sparse subset $S$ of the ground set $E$ such that 
for any $F \subseteq E$ with $|F| \le k$ (i.e., elements in $F$ have failed),
the matroid $M|(S\setminus F)$ contains a 
min-cost basis of $M|(E\setminus F)$. 
Such a subset $S$ is a $k$-$\ftp$ (see Definition~\ref{def:k-ftor}).
Recall the greedy algorithm for min-cost basis of a matroid (see \cite{Goemans}). 
The greedy algorithm takes $O(m\log m)$ time to sort all elements in $E$ 
by cost and it runs at most $m$ independence tests to check if certain sets are in ${\cal I}$,
where  $|E| = m$. Consider Algorithm~\ref{alg1} given below. 

\begin{algorithm}
\caption{Computing a sparse $k$-$\ftp$ for $M = (E, {\cal I})$}
\label{alg1}
\begin{algorithmic}[1]
\State Initialize $B_0 = \emptyset$ and $E_0 = E$.
\For{$i = 1,\ldots,k+1$}
\State Let $E_i = E_{i-1} \setminus B_{i-1}$.
\State Compute a min-cost basis $B_i$ of $M|E_i$ by the greedy algorithm.
\EndFor
\State Return $S = \bigcup_{i=1}^{k+1} B_i$.
\end{algorithmic}
\end{algorithm}

In the first iteration, Algorithm~\ref{alg1} computes a min-cost basis $B_1$ of the matroid $M$.
In the second iteration, it computes a min-cost basis $B_2$ of the matroid $M|E_1$ where $E_1 = E \setminus B_1$. 
This goes on for $k+1$ iterations. That is,
in the $i$-th iteration, where $i \le k+1$,  Algorithm~\ref{alg1} computes a min-cost basis $B_i$ of the matroid
$M|E_i$ where $E_i = E\setminus(B_1\cup\cdots\cup B_{i-1})$. 

Sorting the $m$ elements takes $O(m\log m)$ time and then the algorithm runs
$O(km)$ independence tests (at most $m$ tests per~$B_i$).
We will show the union of bases $B_1,\ldots,B_{k+1}$ is a $k$-$\ftp$.
This implies there is always a $k$-$\ftp$ of size at most $k\!\cdot\!\rank(E)$. Moreover, this is a tight bound (see Remark~\ref{remark:lb}).

\subparagraph{Correctness of Algorithm~\ref{alg1}.} Let us first show correctness for $k = 1$. 
So there is only one faulty element $f$. Assume $f \in B_1$: recall that $B_1$ is a  min-cost basis of $M = (E, {\cal I})$. 
The case when $\rank(E-f) < \rank(E)$ is simple as shown in Observation~\ref{obs1}.
\begin{observation}
    \label{obs1}
If $\rank(E - f) < \rank(E)$ then $B_1 - f$ is a min-cost basis of $M|(E - f)$.    
\end{observation}
\begin{proof}
    Suppose there is some basis $B'$ of $M|(E-f)$ such that $c(B') < c(B_1 - f)$.
    Since $\rank(E - f) < \rank(E)$, we have $f \notin \Span(E-f)$.
    Thus $B'+f$ is a basis of $M$ and moreover, $c(B'+f) < c(B_1)$. This 
    contradicts the fact that $B_1$ is a min-cost basis of $M$.    
\end{proof}

Let us now assume that $\rank(E-f) = \rank(E)$. 
We will show there exists $b \in B_2$ such that $B_1 -f + b$ is a min-cost basis of $M|(E-f)$, where
$B_2$ is a min-cost basis of $M|(E \setminus B_1)$.
Lemma~\ref{lemma1} is proved at the end of this section.

\begin{lemma}
\label{lemma1}
There exists an element $x \in E - f$ such that $B_1 - f + x$ is a min-cost basis of $M|(E-f)$.
\end{lemma}

By Lemma~\ref{lemma1}, we know there exists some ``swap element'' $x \in E - f$ to complete $B_1-f$ to a min-cost 
basis of $E-f$. Our goal now is to show that $B_2$ contains such a swap element. 
If $x \in B_2$ then there is nothing to prove. So let us assume that $x \notin B_2$. 

\begin{lemma}
\label{lemma2}
Let $b_1,\ldots,b_\ell$ be all the elements $b_i \in B_2$ such that $B_2 -b_i + x$ is a basis of $M|(E\setminus B_1)$.
Then $x \in \Span(\{b_1,\ldots,b_\ell\})$.
\end{lemma}
\begin{proof}
    Suppose  $x \notin \Span(\{b_1,\ldots,b_\ell\})$.
    Then $\rank(I + x) = \rank(I) + 1$ where 
    $I = \{b_1,\ldots,b_\ell\}$. So $\rank(I+x) = |I| + 1$, thus $I + x \in {\cal I}$. 
    Consider the independent sets $I+x$ and $B_2$. While $|B_2| > |I+x|$, augment $I+x$ with elements of $B_2 \setminus (I+x)$ so that the final augmented set $I^*$ is a basis of $M|(E\setminus B_1)$. 
    Since $x \in I^* \setminus B_2$, there exists an element $b_t \in B_2 \setminus I^*$.
    Thus $I^* = B_2 - b_t + x$. 
    
    In other words, $B_2 - b_t + x$ is a basis of $M|(E\setminus B_1)$. This means
    $b_t \in B_2$ satisfies the definition for elements to belong to the original $\ell$-set $I = \{b_1,\ldots,b_\ell\}$. 
    But $b_t \in B_2 \setminus I^*$, so $b_t\notin I$, a contradiction. Thus  $x \in \Span(\{b_1,\ldots,b_\ell\})$.
\end{proof}

Let $I = \{b_1,\ldots,b_\ell\}$ (from Lemma~\ref{lemma2}). Lemma~\ref{lemma3} shows $I$ has the desired element $b_i$.
\begin{lemma}
\label{lemma3}
There exists an element $b_i \in I$ such that $B_1 - f + b_i$ is a min-cost basis of $M|(E-f)$.
\end{lemma}
\begin{proof}
By definition of the set $\{b_1,\ldots,b_\ell\}$ (see Lemma~\ref{lemma2}) and the fact that $B_2$ is a min-cost
basis of $E \setminus B_1$, we have $c(B_2) \le c(B_2 -b_i + x)$ for all $i \in \{1,\ldots,\ell\}$. Thus
$c(b_i) \le c(x)$ for all $i \in \{1,\ldots,\ell\}$. Thus we need to show that for some $i \in \{i,\ldots,\ell\}$,
$B_1 -f + b_i$ is a basis of $M|(E-f)$.
Suppose not. This means $b_i \in \Span(B_1 - f)$ for all $i \in \{1,\ldots,\ell\}$.
So $\{b_1,\ldots,b_\ell\} \subseteq \Span(B_1 - f)$. Hence 
 $\Span(\{b_1,\ldots,b_\ell\}) \subseteq \Span(\Span(B_1 - f)) = \Span(B_1 - f)$.
 
 We know from Lemma~\ref{lemma2} that $x \in \Span(\{b_1,\ldots,b_\ell\})$. Thus $x \in \Span(B_1 - f)$,
 a contradiction to the fact that $B_1 -f + x$ is a basis of $E - f$. Thus it must be the case that 
$B_1 -f + b_i$ is a basis of $M|(E-f)$, for some $i \in \{i,\ldots,\ell\}$.
\end{proof}

We prove Lemma~\ref{lemma1} below.

\begin{proof}[Proof of Lemma~\ref{lemma1}]
    Let $\rank(E) = r$. Let $B_1 = \{e_1,\ldots,e_r\}$ where $c(e_1) \le \ldots \le c(e_r)$. 
    Suppose $f = e_i$. Consider the greedy algorithm on $E - e_i$. After selecting elements $e_1,\ldots,e_{i-1}$ and possibly
    $e_{i+1},\ldots,e_j$, a new element $x \notin B_1$ gets selected by the greedy algorithm.
    So $I_j - e_i + x$ is an independent set of size~$j$ where $I_j = \{e_1,\ldots,e_j\}$.
    While $|I_j - e_i + x| < |B_1|$, the set $I_j -e_i +x$ can be augmented with elements of $B_1 \setminus (I_j - e_i + x)$ 
    such that the final augmented set $I^*$ is a basis of $M|(E-e_i)$.
    
    Observe that $I^* = B_1 - e_i + x$.
    This is because $\rank(I_j - e_i + x + e_i) = \rank(I_j + x) = j$ since the
    greedy algorithm on $E$ did not select $x$. Hence $e_i \in \Span(I_j -e_i +x)$. Moreover, all elements of $B_1 - e_1$ are in $I^*$.
    Thus $e_1,\ldots,e_{i-1},e_{i+1}$, $\ldots,e_j,x,e_{j+1},\ldots,e_r$ is a listing of elements of $I^*$ 
    in non-decreasing order of cost.  
    We claim $I^*$ is a min-cost basis of $M|(E-e_i)$. 
    
    Suppose not. Then $c(I^*) > c(B')$ where $B'$ is a min-cost basis of $M|(E-e_i)$.
    So the greedy algorithm---after selecting elements
    $e_1,\ldots,e_{i-1},e_{i+1},\ldots,e_j,x$---selects some elements $y_{j+1},y_{j+2},\ldots,y_r$ 
    where $c(y_{j+1}) \le c(y_{j+2}) \le \cdots \le c(y_r)$
    to form $B'$ where $c(B') < c(I^*)$.  
    Let $t$ be the least index such that $c(y_t) < c(e_t)$.
    
    Let $B'_t = \{e_1,\ldots,e_{i-1},e_{i+1},\ldots,e_j,x,y_{j+1},\ldots,y_t\}$
    be $B'$ restricted to first $t$ elements. 
    Consider the independent sets $I_{t-1} = \{e_1,\ldots,e_{t-1}\}$ and $B'_t$. 
    Since
    $|I_{t-1}| < |B'_t|$, there exists an element $z \in B'_t\setminus I_{t-1}$ 
    such that $I_{t-1} + z \in {\cal I}$. Observe that $z = y_\ell$ for some $j+1 \le \ell \le t$.
    We have $c(I_{t-1}) + c(z) < c(I_t)$ since $c(z) \le c(y_t) < c(e_t)$ for all $z \in B'_t$.
    Thus $c(I_{t-1}+z) < c(I_t)$ where $I_t = \{e_1,\ldots,e_t\}$. This contradicts the fact that $I_t$ is a min-cost size~$t$ independent set in the given matroid $M$: recall that the greedy algorithm finds a min-cost independent set of size $\ell$ for each $\ell \le r$.
    Hence $B_1 - e_i + x$ is a min-cost basis of $M|(E-e_i)$.
\end{proof}

This completes the proof of correctness for $k = 1$. 
\subparagraph{Multiple element failures.}
We know that $B_2$ is a min-cost basis of the restriction of $M$ to $E\setminus B_1$. Upon failure of an 
element $f \in B_1$, a ``swap element'' $f'$ moving from $B_2$ to $B_1$ can be interpreted as the {\em failure} of
an element $f' \in E\setminus B_1$. This is totally analogous to the failure of an element $f \in E$ studied so far. 

Thus $B_2$ has lost an element $f'$ and it follows from Lemmas~\ref{lemma1}-\ref{lemma3} that
there exists a ``swap element'' $f'' \in B_3$ of a min-cost basis in $M|(E\setminus (B_1+f'))$.
Thus recursively, a ``swap element'' $f''$ moving from $B_3$ to $B_2$ can be interpreted as the {\em failure} of
an element $f'' \in E\setminus (B_1 \cup B_2)$. We do this for $k$ levels and thus finally, the min-cost basis $B_{k+1}$
of $E \setminus (\cup_{i=1}^k B_i)$ has lost an element. In other words, this element has moved from $B_{k+1}$ to $B_k$
to replace an element that $B_k$ lost (which moved from $B_k$ to $B_{k-1}$) and so on.

The ground set now is $E' = E-f$ and we have a min-cost basis $B'_1$ of the matroid $M' = M|(E-f)$ and a min-cost
basis $B'_2$ of the matroid $M'|(E'\setminus B'_1)$ and more generally, a min-cost
basis $B'_i$ of the matroid $M'|(E'\setminus \cup_{j=1}^{i-1}B'_j)$ where $1 \le i \le k$. There are
$k-1$ more element failures to process and we have the bases $B'_1,\ldots,B'_k$ with us. Recall that we had $k+1$
bases $B_1,\ldots, B_{k+1}$ at the end of Algorithm~\ref{alg1}.
Thus we have lost one basis $B_{k+1}$---however 
we have processed one element failure (call this failed element $f_1$) and this has resulted in the movement of one element 
from $B_{i+1}$ to $B_i$ for each $i \in \{1,\ldots,k\}$. 

We next process another element failure $f_2$ and this will again result in the movement of one element from $B'_2$ to $B'_1$
and similarly, from $B'_3$ to $B'_2$, and so on. Finally, one element will move from $B'_k$ to $B'_{k-1}$. Thus we will be
left with $k-1$ bases $B''_1,\ldots,B''_{k-1}$ and $k-2$ element failures to process. Hence when it is the turn of
the failure of the $k$-th element $f_k$ to process, we will have two bases $B^{k-1}_1$ and $B^{k-1}_2$ with us. As
seen in Lemmas~\ref{lemma1}-\ref{lemma3}, there will be a ``swap element'' in $B^{k-1}_2$ that can be added to 
$B^{k-1}_1$ so that we have a min-cost basis $B^k_1$ of $E \setminus (f_1\cup f_2\cup\cdots\cup f_k)$.
This completes the proof of correctness of our algorithm. Hence we can conclude Theorem~\ref{thm:k-ftor} stated in Section~\ref{sec:intro}.

\section{Concluding Remarks}
\label{sec:discussion}
Let $G = (V,E)$ be a directed graph on $n$ vertices with a designated root vertex $r$ and non-negative edge costs. 
We gave a simple algorithm to construct a subgraph $H$ of $G$ with $O(n^{3/2})$ edges such that a min-cost arborescence 
in $H-f$ is a 2-approximate min-cost arborescence in $G - f$, for any $f \in E$. Thus any weighted directed graph admits 
a sparse 1-EFT (single edge fault tolerant) subgraph for 2-approximate min-cost arborescence.
The following are interesting open problems.
\begin{enumerate}
    \item Is there a sparse 1-EFT subgraph/preserver for exact min-cost arborescence? Issues involved in extending our construction to a preserver are discussed in the appendix.
    \item Another open problem is on the tightness of our upper bound for 1-EFT subgraph for
    2-approximate min-cost arborescence. Can our bound of $O(n^{3/2})$ be improved?
    \item Another very interesting open problem is to solve the sparse $k$-EFT subgraph problem for exact/approximate 
min-cost arborescence for $k > 1$, i.e., for more than one edge fault.
\end{enumerate}

We also considered the $k$-fault tolerant preserver ($k$-FTP) problem
in a matroid $M = (E, {\cal I})$ for any $k \ge 1$.  We showed there is a subset $S$ of $E$ of size at most $k\!\cdot\!\rank(E)$ 
such that for any 
$F \subseteq E$ with $|F| \le k$, the restriction $M|(S\setminus F)$ contains a min-cost basis of $M|(E\setminus F)$.
There is a lower bound of $k\!\cdot\!\rank(E)$ on $k$-FTP (see Remark~\ref{remark:lb}); thus our bound is tight.

\paragraph{\bf Acknowledgements.} This work was done when DD was a Visiting Fellow at TIFR and supported by the Department of Atomic Energy, Government of India, under project no. RTI4001. TK acknowledges support by the Department of Atomic Energy, Government of India, under project no. RTI4014.


\section*{Appendix}
\subparagraph{\bf Extending Algorithm~\ref{alg0} to find a 1-EFT preserver for min-cost arborescence.}
In order to update $\TT$ to a min-cost arborescence $\TT_f$ in $G-f$, 
edge costs in $G$ should be replaced with {\em effective} edge costs, 
where each vertex reduces the cost of its incoming edge in $\TT$ from the cost of each of its incoming edges in $G - f$.
Then $\cost(\TT_f) = \cost(\TT) + \sum_{e \in \TT_f}\cost'(e)$,
where $\cost'(e)$ is the effective cost of $e$.

\smallskip

Let $\rho'_i$ be the shortest path with respect to effective edge costs from $\XX(v_i)$ to $v_i$. 
 Consider replacing $\rho_1,\ldots,\rho_n$ in Algorithm~\ref{alg0} with $\rho'_1,\ldots,\rho'_n$.
 Let $H' = \TT \bigcup_{i=1}^{n-1}\rho'_i$.

\smallskip

 Note that a min-cost arborescence $A'_f$  in $H'-f$
 need not be a min-cost arborescence in $G-f$.
This is because some vertices (such as $w,x,y$ in Figure~\ref{upwarda}) may have incoming edges with negative effective cost from 
    non-descendants in $A'_f$ (like $v$ in Figure~\ref{upwarda}). These vertices $w,x,y$ are ancestors in $\TT$ of vertices on $\rho'_i$, 
    but they are no longer their ancestors in $A'_f$.

\begin{figure}[h]
\centering

\begin{minipage}{0.45\textwidth}
\centering
\begin{tikzpicture}[scale=1]

\coordinate (r) at (0,5);
\coordinate (u) at (-1.5,-2);
\coordinate (v) at (-1.75,-3);

\draw (r) node[above] {$r$};
\draw (v) node[below] {$v$};

\filldraw (r) circle(2pt);
\filldraw (u) circle(2pt);
\filldraw (-1.75,-3) circle(2pt);

\filldraw (-0.75,1.5) circle(2pt);
\filldraw (-0.95,0.5) circle(2pt);
\filldraw (-1.17,-0.5) circle(2pt);
\filldraw (-1.35,-1.25) circle(2pt);
\filldraw (-0.53,2.5) circle(2pt);
\filldraw (-0.37,3.25) circle(2pt);
\filldraw (-0.15,4.25) circle(2pt);
\filldraw (-0.15,1) circle(2pt);
\filldraw (0.12,0) circle(2pt);
\filldraw (0.55,-1) circle(2pt);
\filldraw (0.65,-2) circle(2pt);
\filldraw (1,-3) circle(2pt);

\draw (r) -- (u);
\draw[densely dashed] (u) -- (v);

\draw[decorate, decoration={snake, amplitude=1mm, segment length=16mm}] 
(-0.53,2.5) -- (1,-3);

\path [very thick,draw,magenta] 
(1,-3) .. controls (0.5,-4) and (-0.75,-4) .. (v);

\path [very thick,draw,cyan] 
(v) .. controls (-2,-1.5) and (-2,-1.5) .. (-1.17,-0.5);

\path [very thick,draw,cyan] 
(v) .. controls (-3,-1.5) and (-3,-1.5) .. (-0.95,0.5);

\path [very thick,draw,cyan] 
(v) .. controls (-4,-1.5) and (-4,-1.5) .. (-0.75,1.5);

\draw (1,-3) node[right] {$u$};
\draw (u) node[right] {$p(v)$};
\draw (-1.17,-0.5) node[right] {$w$};
\draw (-0.95,0.5) node[right] {$x$};
\draw (-0.75,1.5) node[right] {$y$};

\end{tikzpicture}
\subcaption{Let all \textcolor{cyan}{cyan} edges have negative effective cost. The \textcolor{magenta}{magenta} edge $(u,v)$ is expensive, but it needs to be included in any arborescence in $G - (p(v),v)$. Then the \textcolor{cyan}{cyan} edges should also be included in the arborescence to bring down its cost.}
\label{upwarda}
\end{minipage}
\hfill
\begin{minipage}{0.45\textwidth}
\centering
\begin{tikzpicture}[scale=1]

\coordinate (r) at (0,5);
\coordinate (u) at (-1.5,-2);
\coordinate (v) at (-1.75,-3);

\draw (r) node[above] {$r$};
\draw (v) node[below] {$v$};

\filldraw (r) circle(2pt);
\filldraw (u) circle(2pt);
\filldraw (-1.75,-3) circle(2pt);

\filldraw (-0.75,1.5) circle(2pt);
\filldraw (-0.95,0.5) circle(2pt);
\filldraw (-1.17,-0.5) circle(2pt);
\filldraw (-1.35,-1.25) circle(2pt);
\filldraw (-0.53,2.5) circle(2pt);
\filldraw (-0.37,3.25) circle(2pt);
\filldraw (-0.15,4.25) circle(2pt);
\filldraw (-0.15,1) circle(2pt);
\filldraw (0.12,0) circle(2pt);
\filldraw (0.55,-1) circle(2pt);
\filldraw (0.65,-2) circle(2pt);
\filldraw (1,-3) circle(2pt);

\draw (r) -- (u);
\draw[densely dashed] (u) -- (v);

\draw[decorate, decoration={snake, amplitude=1mm, segment length=16mm}] 
(-0.53,2.5) -- (1,-3);

\path [very thick,draw,magenta] 
(1,-3) .. controls (0.5,-4) and (-0.75,-4) .. (v);

\path [very thick,draw,cyan] 
(v) .. controls (-2,-1.5) and (-2,-1.5) .. (-1.17,-0.5);

\path [very thick,draw,cyan] 
(v) .. controls (-3,-1.5) and (-3,-1.5) .. (-0.95,0.5);

\path [very thick,draw,cyan] 
(v) .. controls (-4,-1.5) and (-4,-1.5) .. (-0.75,1.5);

\path [very thick,draw,orange] 
(v) .. controls (0,-1.75) and (0,-1.75) .. (-1.17,-0.5);

\draw (-1.17,-0.5) node[right] {$w$};
\draw (-0.95,0.5) node[right] {$x$};
\draw (-0.75,1.5) node[right] {$y$};
\draw (1,-3) node[right] {$u$};
\draw (u) node[right] {$p(v)$};

\end{tikzpicture}
\subcaption{The \textcolor{orange}{orange} $(w,v)$ edge is cheaper than the \textcolor{magenta}{magenta} edge $(u,v)$. But using the \textcolor{magenta}{magenta} edge will be more optimal as it allows us to include the \textcolor{cyan}{cyan} edges (they have negative effective cost) in our arborescence and this will bring down the entire cost.}
\label{upwardb}
\end{minipage}

\caption{Key differences between 1-EFT preserver for min-cost arborescence and our construction.}
\label{fig:upward}
\end{figure}

If there is a vertex with an incoming edge $e$ in $G-f$ from a {\em non-descendant} in $A'_f$ where $\cost'(e) < 0$ 
(thus $e$ has negative effective cost such as the cyan edges in Figure~\ref{upwarda}) then $A'_f$ is not a min-cost arborescence in $G-f$.
    Hence the subgraph $H'$ needs to be augmented with more edges.
    Furthermore, instead of using the cheapest path $\rho'_i$ from $\XX(v_i)$ to $v_i$,
    it might be more optimal to use another path $q_i$. Using $q_i$---instead of $\rho'_i$---may enable us to use a different set 
    of upward edges (such as the cyan edges in Figure~\ref{upwardb}). So even if $\rho'_i$ is cheaper than $q_i$,
    the cost-saving that we enjoy due to these upward edges more than makes up for the difference in their effective costs. See Figure~\ref{upwardb} for an illustration.

\smallskip

Recall that we used shortest paths $\rho_i$ from $\XX(v_i)$ to $v_i$ in $G - e_i$ (for $i = 1,\ldots,n-1$)
to construct the subgraph $H$ in Section~\ref{sec:approx}. Furthermore, properties of shortest paths were crucial to bound the size of~$H$ by $O(n^{3/2})$. As discussed above, a 1-EFT preserver for min-cost arborescence in $G$
seems to need a new approach that is not via shortest paths. 
Thus it is not clear if there always exists a sparse 1-EFT preserver for min-cost arborescence in weighted directed
graphs and moreover, how to construct such a sparse subgraph efficiently.
\end{document}